\documentclass[preprint,3p,times]{elsarticle}

\usepackage[utf8]{inputenc}
\usepackage{amsthm}
\usepackage{amsmath,amssymb,amsfonts}
\usepackage{mathtools}
\usepackage{graphicx}
\usepackage{caption}
\usepackage{subcaption}

\usepackage[dvipsnames]{xcolor}

\usepackage{url}
\usepackage{enumitem}
\usepackage{listings}
\usepackage{float}
\usepackage[hidelinks]{hyperref}
\usepackage{mathpartir}
\usepackage[ruled,vlined,linesnumbered]{algorithm2e}
\usepackage{placeins}

\usepackage{tikz}
\usetikzlibrary{arrows.meta}

\usepackage{standalone}

\theoremstyle{definition}
\newtheorem{definition}{Definition}

\theoremstyle{plain}
\newtheorem{theorem}{Theorem}

\newtheorem{proposition}{Proposition}
\newtheorem{example}{Example}

\definecolor{dkgreen}{rgb}{0,0.6,0}
\definecolor{nicepurple}{rgb}{0.49,0.34,0.76}

\lstdefinelanguage{Hrebeca}{
    morekeywords={
        softwareclass, physicalclass, changemode, mode, inv, msgsrv, guard,
        statevars, int, float, real, knownrebecs, if, else, self, main, Wire, CAN,
        bool, reactiveclass, true, false, setmode, after, boolean
    },
    otherkeywords={=>,<-,<\%,<:,>:,\#,@},
    sensitive=true,
    morecomment=[l]{//},
    morecomment=[n]{/*}{*/},
    morestring=[b]",
    morestring=[b]',
    morestring=[b]""",
    literate={'}{'}1,
}

\usetikzlibrary{arrows,shapes,positioning,decorations.pathmorphing,shapes,calc}
\newcounter{sarrow}

\providecommand\longrightarrowRHD{\relbar\joinrel\relbar\joinrel\mathrel\RHD}
\providecommand\longrightarrowrhd{\relbar\joinrel\relbar\joinrel\mathrel\rhd}
\makeatletter
\providecommand*\xrightarrowRHD[2][]{\ext@arrow 0055{\arrowfill@\relbar\relbar\longrightarrowRHD}{#1}{#2}}
\providecommand*\xrightarrowrhd[2][]{\ext@arrow 0055{\arrowfill@\relbar\relbar\longrightarrowrhd}{#1}{#2}}
\makeatletter
\newcommand{\overto}[1]{\stackrel{#1}{%
		\overrightarrow{\smash{\,\scriptsize{\phantom{#1}}\,}}}}
\newcommand{\aoverto}[1]{\xrightarrowrhd{#1}}

\newcommand{\toverto}[1]{\stackrel{#1}{\approx\!>}}

\def\int{\mathbb{N}}

\def\VarDec{\mathit{VarDecl}}

\def\CTTS{\mathit{CTTS}}
\def\BCTTS{\mathit{BCTTS}}

\def\RDG{\mathit{RDG}}
\def\cont{\mathit{cont}}

\def\Slice{\mathit{Slice}}
\def\Type{\mathit{Type}}
\def\m{\mathfrak{m}}
\def\sid{\mathit{sid}}
\def\rid{\mathit{rid}}
\def\name{\mathit{name}}

\def\Dom{\mathit{Dom}}

\def\name{\it{vars}}
\def\include{\it{contains}}
\def\now{\it{now}}
\def\ID{\it{ID}}

\def\CName{\it{CName}}
\def\Msg{\it{Msg}}
\def\Var{\it{Var}}

\def\Stmt{\it{Stmt}}

\def\rcv{\it{rcv}}
\def\vars{\it{vars}}
\def\map{\it{map}}
\def\WTraces{\it{WeakTimedTr}}

\def\int{\mathbb{N}}
\def\true{\it{true}}
\def\false{\it{false}}
\def\arv{\it{arv}}
\def\int{\mathbb{N}}
\def\eval{\it{eval}}
\def\expr{\it{expr}}
\def\conf{\it{Config}}
\def\Ins{\it{Ins}}
\def\node{\it{node}}

\def\self{\it{self}}

\def\upd{\it{Upd}}

\def\Valuation{\Sigma}
\def\valuation{\sigma}
\def\Name{\mathit{MName}}
\def\bag{\mathfrak{b}}

\def\body{\mathit{body}}

\journal{Adaptive Programming}

\begin{document}

\begin{frontmatter}


\title{Behavioral Analysis of Timed Actors using Syntactic Slice Equivalence} 

\author[ut]{Ali Ataollahi}
\author[ut]{Fatemeh Ghassemi}
\author[itu,uio]{Eduard Kamburjan}
\author[mdu,kth]{Marjan Sirjani}

\affiliation[ut]{
    organization={University of Tehran},
    city={Tehran},
    country={Iran}
}

\affiliation[itu]{
    organization={IT University of Copenhagen},
    city={Copenhagen},
    country={Denmark}
}

\affiliation[uio]{
    organization={University of Oslo},
    city={Oslo},
    country={Norway}
}

\affiliation[mdu]{
    organization={Mälardalen University},
    city={Västerås},
    country={Sweden}
}
\affiliation[kth]{
    organization={KTH},
    city={Stockholm},
    country={Sweden}
}
\begin{abstract}
Tiny twins are compact behavioral models derived from timed actor models for
selected observable messages. When a source model evolves, regenerating its
tiny twin requires state-space exploration and reduction even if the relevant
behavior is unchanged. We present a static analysis for Timed Rebeca that
compares backward slices of Rebeca dependence graphs for a given set of
observable message names. For the Zeno-free fragment with \texttt{after}
annotations and no \texttt{delay} statements, we prove that slice equivalence
implies weak timed bisimulation under the selected observations. This preserves
observable actions and total elapsed time across internal transitions, allowing
the existing tiny twin to be reused. 
We evaluate the implementation on ten benchmark models paired with revisions
that preserve their observable slices. The cost of static comparison depends on
the size of the graphs representing source statements and their dependencies,
while tiny-twin generation depends on the number of reachable states and
transitions. This difference is reflected in the measurements: static comparison
takes less than one second using tens of megabytes of memory, while tiny-twin
generation can take over an hour and use hundreds of gigabytes.
\end{abstract}

\end{frontmatter}

\section{Introduction}\label{sec:introduction}
Programs evolve as requirements change, faults are corrected, and components
are replaced. A change need not affect every use of a program or its models.
This matters when a general model supports several tasks, each concerned with
only part of its behavior. Models developed for early verification and
validation can serve such later uses~\cite{DBLP:journals/tosem/CederbladhCS24}.
When the source model changes, we therefore need to determine which specialized
models remain valid and which must be regenerated.

A \emph{tiny twin} is a compact behavioral model obtained by reducing a timed
actor model's state space with respect to selected observable messages.
Tiny twins have been used to detect cyber-attacks at
runtime~\cite{DBLP:journals/jpdc/MoradiPAS24}, but the underlying reduction is
useful whenever only a particular observational context matters. Generating a
tiny twin requires state-space exploration and reduction, even when a source
change leaves that context unaffected. A static check that establishes this
fact can avoid repeating these computations.

We present such a check for Timed
Rebeca~\cite{DBLP:journals/scp/ReynissonSACJIS14}. Given an original model
$\mathcal{M}$, a revision $\mathcal{M}'$, and a set $O$ of observable message
names, we construct their Rebeca dependence graphs (RDGs) and compare the
backward slices induced by $O$. These slices represent source statements and
their dependencies; tiny twins represent behavior. The analysis must account
for asynchronous communication~\cite{SabouriSirjani2010Actor} and the timing
dependencies that can affect an observation~\cite{sliced}. We consider the
Zeno-free fragment with \texttt{after} annotations and no \texttt{delay}
statements: an execution cannot perform infinitely many actions within bounded
time. For this fragment, we prove that equivalent observable RDG slices imply
weak timed bisimulation under $O$, preserving observable actions and total
elapsed time across internal transitions. Acceptance therefore permits reuse
of the existing tiny twin. Rejection is inconclusive: different slices may
still describe the same observable behavior.

Our contributions are (1) a static slice-equivalence criterion for establishing that
a revision of a Timed Rebeca model preserves the behavior observable through
$O$; (2) a proof of its soundness, showing that acceptance guarantees weak
timed bisimulation and permits reuse of the existing tiny twin; and (3) an
evaluation on ten benchmark models comparing the cost of
the static decision with tiny-twin regeneration through state-space
exploration and reduction.
Section~\ref{sec:preliminaries} reviews the behavioral equivalence relations we exploited in our behavioral analysis and the classic program dependence graph. We introduces the timed Rebeca syntax, semantics, and its tiny twins in Section \ref{sec:timedRebeca}.
Section~\ref{sec:rdg} defines the analysis and establishes soundness.
Section~\ref{sec:case} presents the smart-home case study.
Section~\ref{sec:experiments} presents the tools and experiments,
Section~\ref{sec:discussion} examines how the comparison distinguishes causality
loops under evolution, and Section~\ref{sec:related} discusses related work.

\section{Preliminaries}\label{sec:preliminaries}
We explain basic notions of behavioral equivalences that are used in our tiny twin definition. We also review the classic program dependence graph.

\subsection{Equivalence Relation}\label{subsec::equiv}
We use weak timed bisimulation~\cite{BrengosPeressotti2019}, the timed extension of weak bisimulation~\cite{Milner1989}, to compare two BCTTSs that accumulate time over a path while ignoring the internal actions.
%
Let $\Rightarrow$ be the reflexive and transitive closure of $\tau$-transitions. We define $\xRightarrow{a}=\Rightarrow\overto{a}\Rightarrow$, and $\toverto{t}=\Rightarrow\overto{t_1}\Rightarrow\overto{t_2}\Rightarrow\ldots\overto{t_n}\Rightarrow$ where $t=t_1+t_2+\ldots+t_n$. 

\begin{definition}[Weak Timed Bisimulation]
\label{def:twt-equivalence}
Let $(S,\rightarrow,s_0)$ be a timed transition system. 
A binary relation $\mathcal{R}\subseteq S\times S$ is a weak timed simulation if $s\,\mathcal{R}\, u$ implies: 
\begin{enumerate}
 \item if $s\xRightarrow{a}s'$, then there exists a state $u'$ such that $u\xRightarrow{a}u'$ and $s'\,\mathcal{R}\, u'$;
 \item if $s\toverto{t}s'$, then there exists a state $u'$ such that $u\toverto{t}u'$ and $s'\,\mathcal{R}\, u'$.
\end{enumerate}
$\mathcal{R}$ is a weak timed bisimulation if $\mathcal{R}$ and $\mathcal{R}^{-1}$ are weak timed simulations. Two states $s$ and $u$ are weakly timed bisimilar, denoted by $s\simeq_{wtb}u$, if $s\,\mathcal{R}\, u$ for some weak timed bisimulation relation $\mathcal{R}$. 
\end{definition}

Two timed transition systems $\mathcal{T}_1=(S_1,\rightarrow_1,{s_0}_1)$ and $\mathcal{T}_2=(S_2,\rightarrow_2,{s_0}_2)$ are  weak timed bisimilar, denoted by
$\mathcal{T}_1\simeq_{\mathrm{wtb}}\mathcal{T}_2$, if there is a weak timed bisimulation relation $\mathcal{R}$ such that ${s_0}_1 \,\mathcal{R}\,{s_0}_2$. 

Weak timed trace equivalence~\cite{BrengosPeressotti2019} compares states by the sequences of observable actions and elapsed times they can exhibit: internal actions are ignored, and consecutive time delays, possibly separated by internal actions, are summed.

\begin{definition}[Weak Timed Trace Equivalence]\label{Def::weaktrace}
Given $\mathcal{T}=(S,\rightarrow,{s_0})$, for each state $s\in S$, $\WTraces(s)\in (\Name^* \times \int)^*\cup (\int \times \Name^* )^*$ is a set defined inductively :\begin{itemize}
\item $\epsilon\in\WTraces(s)$ ;
\item $\rho \in\WTraces(s)$, if there is $s'\in S$ such that $s\overto{\tau}s'$ and $\rho\in\WTraces(s')$;
\item $a\frown \rho \in\WTraces(s)$, if there is $s'\in S$ such that $s\overto{a}s'$, where $a\in\Name$, and $\rho\in\WTraces(s')$;
\item $(t_1+t_2)\frown \rho \in\WTraces(s)$, if there is $s'\in S$ such that $s\overto{t_1}s'$, where $t_1\in\int$, and $t_2\frown\rho\in\WTraces(s')$;
\end{itemize}
Two states $s_1$ and $s_2$ are called weakly timed equivalent if and only if their traces are weakly timed trace equivalent:
\[
s_1\simeq_{wtt} s_2 \text{ iff. }\WTraces(s_1)=\WTraces(s_2)
\]
\end{definition}

\subsection{Program Dependence Graph}\label{subsec::PDG}
Various definitions of program dependence representations have been introduced in the literature as an intermediate representation of a program for the intended static analysis. The common feature among them is the explicit representation of data and control dependence \cite{FerranteOW87}.

The \emph{program dependence graph} for a program $P$, denoted by $G_P=(V,E_C\cup E_D)$, is a directed graph whose vertices are connected by several kinds of edges. Restricting a program to the statements of Timed Rebeca, the vertices represent the assignments, send statements, and control predicates (conditional statement) that occur in the program. In addition, $G_P$ includes an \emph{entry} and \emph{exit} vertices. 

The edges of $G_P$ represent either a \emph{control} ($E_C$) or \emph{data} ($E_D$) dependence between program components. Control dependence edges are labeled either $\it true$ or $\it false$. The source of a control dependence edge is always the entry or a predicate vertex. A control dependence edge $(v_1,b,v_2)\in E_C$, denoted by $v_1\overto{b}_c v_2$ means that during execution, whenever the predicate represented by $v_1$ is evaluated to the value matching the label $b$, then the program component represented by $v_2$ will be executed. Restricting a program to the statements of Timed Rebeca, the control dependence edges reflect a program's nesting structure; $(v_1,b,v_2)\in E_C$ iff one of the following holds:
\begin{itemize}

\item $v_1$ is the entry vertex, and $v_2$ represents a component of $P$ that is not subordinate to any control predicate. Such edges are labeled $\it true$;
\item $v_1$ is a control predicate of a conditional statement, and $v_2$ represents a component of $P$ immediately subordinate to the control construct whose predicate is represented by $v_1$. The edge is labeled by $\it true$ or $\it false$ depending that $v_2$ occurs in the then or else branch, respectively.
\end{itemize}
A data dependence edge $(v_1,v2)\in E_D$ means that the program's computation might be changed if the relative order of the components represented by $v_1$ and $v_2$ were reversed. $(v_1,b,v_2)\in E_D$ iff one of the following holds:
\begin{itemize}
\item $v_1$ defines variable $x$, and $v_2$ uses  $x$, while the control can reach from $v_1$ to $v_2$ via an execution path along which there is no intervening definition of $x$;
\item $v_1$ and $v_2$ both defines the same variable $x$, and are in the same branch of any conditional statement that encloses both of them their relative order of execution is from $v_1$ to $v_2$. Furthermore, there exists a program component $v_3$ using the variable $x$ such that $(v_1,v_3)\in E_D$ and $(v_2,v_3)\in E_D$.
\end{itemize}

The relationship between a program's PDG and the program's execution behaviour has been investigated in \cite{HorwitzPR88a}, as stated in the following theorem.
\begin{theorem}[Strong Equivalence]\label{The::StrongEquvalence}
If the PDGs of two programs $P$ and $Q$ be isomorphic, then for any state $\valuation$ both diverge when initiated on $\valuation$ or they both halt with the same final values for all variables.
\end{theorem}
The isomorphism is not merely structural graph isomorphism graph. The corresponding nodes must represent equivalent operations/statements, and the relevant semantic information has to be preserved.

Treating the body of messages servers of Timed Rebeca models as a simple program, their PDGs can be derived similarly. As there is no loop statements in Timed Rebeca, message servers with isomorphic PDGs, their executions within two rebeca models result in the same final values for all variables. We remark that each send statement can be conceived as a set of assignment; assignment to the actual message parameters, and assignment the message to the bag of receiver. 

\begin{example}
Consider the following programs that their resulting PDGs are isomorphic:
\[\begin{array}{ll}
{\it prog}_1: ~~x=4;~y=y+3;~z=x;~ \mathsf{if}~(x+y>4)~\{a=z;\}\\
{\it prog}_2:~~ y=y+3;~x=4;~z=x;~ \mathsf{if}~(x+y>4)~\{a=z;\}\\
{\it mprogsg}_3:~~ x=4;~z=x;~y=y+3;~ \mathsf{if}~(x+y>4)~\{a=z;\}
\end{array}\]There is a data-dependence edge from the node $x=4$ to $z=x$, and from the nodes $x=4$ and $y=y+3$ to the conditional statement, i.e., the predicate $x+y>4$. 
\end{example}
\section{Timed Rebeca and Tiny Twins}\label{sec:timedRebeca}
We use Timed Rebeca~\cite{DBLP:journals/scp/ReynissonSACJIS14} as the language describing timed actors. We explain its syntax, and define the coarse semantics. 
At its core, an actor system consists of a number of actors~\cite{gul}, which communicate with each other solely using asynchronous messages. Each actor encapsulates its internal state fully, i.e., no other actor can access its state directly, and handles messages using preemption-free handlers: Once a message is received by an actor, it is put into the actor's queue and later handled by a message handler. Each message handler runs without interruption from other handlers until it eventually terminates.

The time model is that there is a global clock, and each actor can \emph{delay} the delivery of a sent message by a given amount of time units. The global clock is advanced whenever no handler is running, and at least one message is delaying its delivery. The global clock is advanced by the minimal time needed to deliver the message. Note that internal delays can be encoded using delay of message delivery.
In the remainder of this paper, we consider the \emph{Zeno-free} fragment, where there is are infinitely many messages in finite, strictly positive time.

\subsection{Syntax}
Rebeca \cite{DBLP:journals/fuin/SirjaniMSB04} and its extensions use an object-based approach with Java-like syntax to express actor systems.
Semantics of Rebeca is  close to original actor model and Rebeca is the first actor-based language with a model checking support~\cite{DBLP:journals/csur/BoerSHHRDJSKFY17,DBLP:conf/fsen/KhamespanahSK23}. The main feature is that sending messages are mapped to asynchronous method calls and message handlers to methods , and delays for delivering messages is annotated at call site of method calls using an \texttt{after} construct. 

\begin{figure}[h]
\begin{center}
\fbox{\parbox{0.75\columnwidth}{\vspace{-2mm}
{\footnotesize
    \begin{align*}
    \mathrm{Model} &\Coloneqq~\langle{Class}\rangle^+~\mathrm{Main}~\\
    \mathrm{Main} &\Coloneqq ~\mathsf{main}~ \{\mathrm{RebecInstance}^*\} \\
    \mathrm{RebecInstance} &\Coloneqq ~\mathrm{C}~ \mathrm{r}~(\mathrm{\langle  r\rangle}^* )~\colon(\mathrm{\langle c\rangle}^*)\\
    \mathrm{Class} &\Coloneqq~ \mathsf{reactiveclass} ~\mathrm{C}~ \{\mathsf{KnownRebecs}~ \mathsf{Vars}~ \mathsf{MsgSrv}^* \}\\
    \mathrm{KnowRebecs} &\Coloneqq~ \mathsf{knownrebecs}~\{\mathrm{VarDcl} \}\\
    \mathrm{Vars} &\Coloneqq \mathsf{statevars}~ \{\mathrm{VarDcl} \}\\
    \mathrm{VarDcl} &\Coloneqq \langle\mathrm{T}~ \mathrm{v}\rangle^* ;\\
    \mathrm{MesgSrv} &\Coloneqq \mathsf{msgsrv} ~\mathrm{m}~(\mathrm{VarDcl})~\{\mathrm{Stmt}^* \}\\
    \mathrm{Stmt}
    &\Coloneqq~
    \mathrm{Assign};~|~\mathrm{Send};~|~\mathsf{if}(\mathrm{Expr})~\mathrm{MSt}~[\mathsf{else}~\mathrm{MSt}]~|~ {\color{blue}\mathsf{delay}(\mathrm{Expr});}\\
    \mathrm{Assign} & \Coloneqq~\mathrm{v}\mathsf{=}\mathrm{Expr}~|~ ~\mathrm{v}\mathsf{=?(}\mathrm{Expr}\langle,\mathrm{Expr}\rangle^+\mathsf{)}\\
    \mathrm{Send}
    &\Coloneqq~\mathrm{r.m(\langle Expr\rangle ^*)}
    [\mathsf{after}~ \mathrm{Expr}]\\
    \mathrm{MSt}
    &\Coloneqq~\{\mathrm{Stmt}^*\}~|~ \mathrm{Stmt} 
    \end{align*}
}}}
\end{center}
\vspace{-2mm} 
\caption{
Timed Rebeca Syntax: 
$\langle~\rangle$ denotes meta-parentheses, 
$[~]$ indicates optional syntax, and $\langle~\rangle^\ast$ denotes comma-separated lists; identifiers \texttt{C}, \texttt{T}, \texttt{m}, 
and \texttt{r} represent class, type, method, 
and rebec names, respectively.}
\label{fig:timedSyntax}
\end{figure}

The syntax of Timed Rebeca is given in Fig.~\ref{fig:timedSyntax}. Each Timed Rebeca model consists of a set of reactive classes and a \texttt{main} block. A reactive class specifies the
type of its runtime instances, called \emph{rebecs}; it declares state variables,
known rebecs, a constructor, and message servers. Known rebecs specify the set of rebecs to which instances of the class can send messages. State variables are local to a rebec. Message servers, or methods, declared with the keyword \texttt{msgsrv}, specify the behavior of a rebec when handling incoming messages using variable assignment, nondeterministic assignment, conditional statement, delay statement, and message sending. The \texttt{main} block creates the rebecs in the model and supplies their known rebecs and constructor
arguments.  

Communication between rebecs is therefore exclusively by asynchronous, non-blocking message passing to known rebecs. Each rebec owns a message bag rather than a FIFO queue. When it is idle, it
takes a message from its bag and executes the corresponding message
server. A rebec is single-threaded, so two of its message servers cannot execute simultaneously; message-server steps of different rebecs may nevertheless be
interleaved. A rebec may send messages to itself to express recurrent behavior. 

The send statement \texttt{r.m()~after~t} denote that the message is delivered to \texttt{r} after \texttt{t} times. 
The construct \texttt{after~t} proceeding a send statement models the communication delay. The \texttt{delay(t)} statement is useful for modeling the computation delay, which resumes the execution of the rebec after $t$ times. We consider the fragment of Timed Rebeca without the \emph{delay} statement which simplifies the representation of time progress in the semantic model. Therefore, each message server is executed without being paused and resumed later. So, each message server is executed at its message's arrival time. 
From now on, Timed Rebeca refers to this fragment. 

Nondeterministic assignments are part of the Timed Rebeca fragment considered in this paper and are used to model choices in the
environment. A nondeterministic assignment has the form $x := ?(e_1,\ldots,e_n)$ leading to a nondeterministic behavior.  


\subsection{Operational Semantics}
The semantics of Timed Rebeca are expressed in terms of a Timed Transition System (TTS).
%
We introduce a few notations and auxiliary functions that are used in describing the formal model and semantics of Timed Rebeca.

\paragraph{Notations}
Given a set $A$, we use $A^{*}$ to show a set of all finite sequences on the elements of $A$. For a sequence $a\in A^{*}$ with the length of $n$, the symbol $a_i$ denotes the $i^{th}$ element of the sequence, where $1 \leq i \leq n$. We write $x\in a$ to denote that there exists $i$ such that $a_i=x$. We use $a\setminus x$ to show a sequence achieved from removing an occurrence of $x$ from the sequence. An empty sequence is denoted by $\epsilon$. We write $a\frown x$ to express a sequence achieved from appending $x$ to the end of a sequence. 
For a function $f:X \rightarrow Y$, we define $f[\alpha \mapsto \beta]:X\cup\{\alpha\}\rightarrow Y\cup\{\beta\}$ to be the function with $f[\alpha \mapsto \beta](x)=f(x)$ for all $x\in X\setminus\{\alpha\}$ and $f[\alpha \mapsto \beta](\alpha)=\beta$. We use $\Dom(f)$ to denote the domain of the function $f$.

\paragraph{Basic Definitions}
Let $\ID$ be the set of \emph{rebec identifiers}, $\CName$ the set of class names, $\Name$ the set of message server names. 
Given the set $\Var$ of variables with values from $\textit{Val}$ (including boolean, $\int$ values), a \emph{valuation} $\valuation: \Var \rightarrow \textit{Val}$ is a function that maps a value to each variable from its domain $\Dom(\valuation)=\Var$. The set of all valuations is denoted by $\Valuation$. Let $\valuation \downharpoonright
v$ denote the valuation whose domain is restricted to the variables of $v$, i.e., $\Dom(\valuation\downharpoonright v)=\Dom(\valuation)\cap v$. For a valuations $\valuation_1: \Var_1 \rightarrow \textit{Val}$ and $\valuation_2: \Var_2 \rightarrow \textit{Val}$ over disjoint domains $\Var_1\cap\Var_2=\emptyset$, we define their \emph{union} $\valuation_1\uplus \valuation_2$ to assign $\valuation_1(x)$ to each $x\in\Var_1$, and $\valuation_2(x)$ to each $x\in\Var_2$. Two valuations $\valuation_1$ and $\valuation_2$ are called equal if $\Dom(\valuation_1)=\Dom(\valuation_2)$ and $\forall v\in\Var\cdot \valuation_1(v)=\valuation_2(v)$. 
Let $\VarDec=\Var\times\Type$ denote a variable declaration, where $\Type$ is the set of  primitive types including boolean and int. We use $\name(\nu)$ to denote the variable of the declaration $\nu$.

A \emph{message} $\m$ is a tuple $\m = (\sid,\rid,m,par,\arv)$, where $\sid\in \ID$ is the sendeing rebec identifier, $\it rid$ is the receiving rebec identifier, 
$m\in\Name$ is the name of the message and its handling message server, $par\in \Sigma$ is the parameter valuation defining the values of the parameters, and $\arv\in \int$ is its arrival time. 
We use the notation $\m.\sid$, $m.\rcv$, and $\m.\arv$ to denote the sender, receiver, and the arrival time of a message. We also use $\m.\name$ to represent the name of a message.
Let $\Msg$ denote the set of all messages. 
Each rebec has a bag in which it stores all the messages that it has received. We represent the bag content as the sequence of messages in it, i.e., $\bag=\langle \textit{msg}_1, \textit{msg}_2, \dots , \textit{msg}_n \rangle \in \Msg^*$ for some $n \in \int$.  Let $O\subseteq \Name$ be a set of names. We use $\bag \downharpoonright O$to denote a bag containing those messages of the bag whose names belong to $O$ and $|\bag |$ indicate the size of the bag.   

A \emph{Timed Rebeca Model} is a concise representation of the state of a Timed Rebeca program, including its class table $C$, as well as the set of rebecs $\Ins$ and the current state $\conf$. Note that $\Ins$ is constant.

\begin{definition}[Timed Rebeca Model]
\label{def:timed-rebeca-model}
A Timed Rebeca model $\mathcal{M}$ is defined by the triple $(C,\Ins,\conf)$, where the components are as follows.
\begin{itemize}
    \item $C=(C_i)_{i \in I}$ is the set of reactive class declarations for some index set $I$, Each reactive class is described by $C_i=(c_i,V_i,M_i,K_i)$, where $c_i\in\CName$ is the class name, $V_i\subseteq \VarDec$ is the set of state variable declarations, $M_i\subseteq \Name\times \VarDec\times\Stmt^*$ is its set of message
servers declarations defined by a triple like $(m,\nu,st)$, and $K_i\subseteq \Var\times C$ is the set of known rebec declarations, 
\item 
    $\Ins\subseteq \ID\times \CName$ is the set of rebec instances, and
    \item $\conf\subseteq \ID\times (\Var\rightarrow \ID)
$ is the configuration expressing the set of known rebecs for each rebec. 
\end{itemize}  
\end{definition}

We use $\body(x, m)$ 
to denote the body 
of the message server with the name $m$ of the rebec with the identifier $x$. 
We write $\vars(x)$ to express the state variables of a rebec with the identifier $x$:
\[
\begin{array}{l}
(x,c_i)\in \Ins,\wedge\, C_i=(c_i,V_i,M_i,K_i) \Rightarrow \vars(x)=\{v_i\,\mid (v_i,T_i)\in V_i\}\\
(x,c_i)\in \Ins,\wedge\, C_i=(c_i,V_i,M_i,K_i) \,\wedge \, (m,\nu,st)\in M_i\Rightarrow \body(x,m)=st 
\end{array}
\]

We define the semantic model of a Timed Rebeca model in terms of timed transition systems. Each global state is defined by the local states of rebecs and the logical global time. The global state changes when a rebec handles one massage from its bag. As rebecs do not have any shared variables, instead of interleaving the executions of message server's statements of two rebecs, we coarsely handle message servers. 

\begin{definition}[Coarse Timed Transition System]
\label{def:rtts}
The Coarse Timed Transition System of a Timed Rebeca model $\mathcal{M}$ is denoted by $\operatorname{CTTS}(\mathcal{M})=(S,\rightarrow,s_0)$, where \begin{itemize}
\item $S\subseteq (\ID\rightarrow \Valuation\times \Msg^* \times \Stmt^* )\times \int$ is the set of global states. Each global state is a pair of a mapping from the rebec identifiers to their local states, and the global time. We use $\now(s)$ to indicate the global time of a state, and $\map(s)$ to express the mapping of a state. The local state of a rebec is represented by the triple $(\valuation,\bag,st)$, where $\valuation$ assigns values to the state variables, $\bag$ is its message bag, and $st$ is the sequence of remaining statement that the rebec is going to execute.
\item $\rightarrow\subseteq S\times (\Name\cup\int) \times S$ is the transition relation derived by the SOS rules in Tab. \ref{Tab:sos-rules}, and
\item $s_0$ is the initial state with logical time zero. Each rebec local variables are initialized to the default values (zero for int and false for boolean) and their bags contain the messages triggered by executing the constructors. 
\end{itemize}
\end{definition}


Intuitively, there are two kinds of transitions in CTTS~\cite{KhamespanahSVK16}: \begin{enumerate}
  \item A \emph{taking-event} action removes a message name $m$ from the bag of an
  idle rebec when $m.\arv=\mathit{now}(s)$, and executes the
  corresponding message server, as defined by the rule $\it TakeMessage$ in Tab. \ref{Tab:sos-rules}. We coarsely execute the body of the message server as indicated by $(\valuation\uplus par,\body(x,m),\epsilon),t\aoverto{}^*(\valuation',\epsilon,out),t$ which execute the body of the message server using the rebec valuation extended with the message parameter valuation. The messages sent during the message server to the other rebecs are gathered in $out$. We dispatch messages in $out$ to their receivers by using $\upd$ as defined below: 
  \[\begin{array}{l} 
\upd(s,\epsilon) =s \\
\upd(s,m_1 m_2 \ldots m_n) = \upd(s[x\mapsto(\valuation,\bag\frown m_1,st)],m_2 \ldots m_n)~~\text{where $m_1.\rcv=x$ and $s(x)=(\valuation,\bag,st)$}
\end{array}\]


  \item A \emph{time} action occurs only when no message can be taken. 
  It advances the clock to the earliest arrival
  time of a pending message, as expressed by the rule $\it TimeProg$ in Tab. \ref{Tab:sos-rules}. 
\end{enumerate}

\begin{table}[h]
    \centering
        \caption{Operational rules of the Timed Rebeca fragment; The rule names appended with $*$ denote the transitions of rebec local states: the  counterpart $\textsc{CondF}*$ for the conditional statement has been removed for brevity. 
        }
    \label{Tab:sos-rules}
    \begin{tabular}{|c|}
    \hline
$\inferrule*[left=(Assign*)]{\ }{\Big(\valuation,v:=\expr;st,out\Big),t~\aoverto{\tau}~\Big(\valuation\big[v\mapsto\eval(\expr,\valuation)\big],st,out\Big),t}$ 
\\[2mm]
$\inferrule*[left=(nonDeteAssign*)]{
0<i\le n}{\Big(\valuation,v:=(\expr_1,\ldots,\expr_n);st,out\big), t~\aoverto{\tau}~\Big(\valuation\big[v\mapsto\eval(\expr_i,\valuation)\big],st,out\Big), t}$\\[2mm]
$\inferrule*[left=(CondT*)]{\eval(\expr,\valuation)=\true}{\Big(\valuation,{\it if}~ (\expr)~st_1~{\it else} ~st_2;st,out\Big),t~\aoverto{\tau}~\Big(\valuation,st_1;st,out\Big),t}$\\[2mm]
$\inferrule*[left=(CondF*)]{\eval(\expr,\valuation)=\false}{\Big(\valuation,{\it if}~ (\expr)~st_1~{\it else} ~st_2;st,out\Big),t~\aoverto{\tau}~\Big(\valuation,st_2;st,out\Big),t}$\\[2mm]
$\inferrule*[left = (Send*)]{\ }{\Big(\valuation,y.m(\expr_1)~\mathtt{after}~(\expr_2);st,out\Big),t~\aoverto{\tau}~\Big(\valuation,st,out \frown \big(\self,y,m,\eval(\expr_1,\valuation), t+\eval(\expr_2,\valuation)\big)\Big),t}$\\[3mm]
$\inferrule*[left=(TakeMessage)]{\mathfrak{s}(x)=(\valuation,\bag,\epsilon)\\ (y,x,m,par,t)\in\bag\\ \big(\valuation\uplus par,\body(x,m),\epsilon\big),t~\aoverto{}^*~\big(\valuation',\epsilon,out\big),t}{\Big(\mathfrak{s},t\Big)~\overto{m}~\Big(\upd\big(\mathfrak{s}[x\mapsto (\valuation'\downharpoonright \vars(x),\bag\setminus m,\epsilon)],out\big),t\Big)}$  \\[2mm]
$\inferrule*[left=(TimeProg)]{\forall x\in\ID\cdot \mathfrak{s}(x)=(\valuation,\bag,\epsilon)\wedge \forall \m\in\bag\cdot\m.\arv>t\\ tp={\it min}(\{\m.\arv\,\mid\,\exists x\in\ID\cdot \mathfrak{s}(x)=(\valuation,\bag,\epsilon)\wedge  \m\in \bag \})}{\big(\mathfrak{s},t\big)\overto{tp-t}\big(\mathfrak{s},tp\big)}$\cr
\hline
\end{tabular}

\end{table}

The rules $\textsc{Assign}$, $\textsc{NonDeteAssign}$, and $\textsc{CondT}$ define how the local state of a rebec is changed upon the execution of assignment, nondeterministic assignment, and conditional statements, respectively. The rule $\textsc{Send}$ indicates that upon sending a message, its arrival time is computed and inserted into $out$ while the rule $\textsc{TakeMessage}$ will insert it into the recipient's bag. The message is taken from a bag when its arrival time is equal to the global time. Messages arriving at the same time to a rebec are handled nondeterministically in different orders (but at the same logical time).

The absolute logical time can make the CTTS infinite even when a reactive model
revisits the same behaviour. Afra therefore applies shift equivalence which
identifies states that have identical non-timing information and whose timing
components differ by a common offset. This yields a bounded CTTS (BCTTS) which is equivalent with the semantic model called CRTTS in the work of Moradi et al.~\cite{DBLP:journals/jpdc/MoradiPAS24}. We use $\BCTTS(\mathcal{M})$ to express the bounded CTTS derived from the timed rebeca model $\mathcal{M}$. 

We generalize the shift-equivalent given in \cite{DBLP:conf/agere/KhamespanahSKSI12,KhamespanahSSKI15} by introducing two parameters $M\subseteq\Name$ and $V\subseteq \Var$. In this parametrized version, only values of variables in $V$ are compared, and only the  messages of bags with a name from $M$ are considered. 
\begin{definition}[$M$-$V$ shift-equivalence relation]\label{Def::shift}
Two states $s$ and $s'$ are called $M$-$V$\emph{shift-equivalent}, denoted by $s\simeq_{V}^{M}s'$, if there exists $\delta\in\int$ such that (1) $\now(s)=\now(s')+\delta$, and (2) for all the rebecs with the identifier $x\in\ID$ where $\map(s)(x)=(\valuation,\bag,st)$ and $\map(s')(x)=(\valuation',\bag',st')$, the following conditions hold:
\begin{enumerate}
    \item $\valuation\downharpoonright V=\valuation'\downharpoonright V$,
    \item $|\bag \downharpoonright M|=|\bag' \downharpoonright M|$,
    \item $\forall m\in \Msg\cdot (sid,rid,m,par,\arv)\in\bag\downharpoonright M \Leftrightarrow (sid,rid,m,par,\arv+\delta)\in\bag'\downharpoonright M.$
\end{enumerate}

\end{definition}

\noindent Two states are shift-equivalent as defined in \cite{DBLP:conf/agere/KhamespanahSKSI12,KhamespanahSSKI15} when they are $\Name$-$\Var$ shift-equivalent.

\subsection{Tiny Twins \label{subsec::tiny}}
We define tiny twins~\cite{DBLP:journals/jpdc/MoradiPAS24}, again following Moradi et al. Given a Timed Rebeca model $\mathcal{M}$, a tiny twin $\mathcal{T}$ is a slice that is weakly timed equivalent to $\mathcal{M}$ with respect to a subset of observable messages $O\subseteq \Name$. This means that $\mathcal{T}$ cannot be distinguished from $\mathcal{M}$ using timed behavior over $O$.
The operator $\mathcal{A}_O$, parameterized with $O\subseteq \Name$, abstracts the transitions with a label $m\not\in O$ to $\tau$. 

The abstracted semantic model is reduced by weak timed bisimulation $\simeq_{wtb}$, as defined in Section \ref{subsec::equiv}. We restrict to Zeno-free Timed Rebeca models: every execution contains only
finitely many instantaneous action transitions within any bounded interval of logical time. In particular, an execution cannot perform infinitely many
internal actions without time progressing. This assumption ensures that a
finite elapsed-time segment contains only finitely many internal steps. 

\begin{definition}[Tiny Twin]\label{Def::tiny}
Tiny Twin from a Timed Rebeca model $\mathcal{M}$ with respect to the observable messages $O$ is formally achieved from $\mathcal{A}_O(BCTTS(\mathcal{M}))/\simeq_{wtb}$. 
\end{definition}



Two Timed Rebeca models are considered behaviorally equivalent
when the initial states of their corresponding BCTTSs are weakly timed bisimilar.

\begin{proposition}\label{Pro::abstraction}
Let $\mathcal{T}_1$ and $\mathcal{T}_2$ be two BCTTSs such that $\mathcal{A}_O(\mathcal{T}_1)\simeq_{\mathrm{wtb}}\mathcal{A}_O(\mathcal{T}_2)$ for an arbitrary set $O\subseteq \Name$. Then, $\mathcal{T}_1\simeq_{\mathrm{wtb}}\mathcal{T}_2$.
\end{proposition}

To compact the tiny twins by removing all the $\tau$-transitions, we can compute its modulo weak timed trace equivalence, which also combines time elapse transitions separated by $\tau$-transitions. As our BCTTS models are generally deadlock-free, such a reduction can provide more insight into the model
\section{{Dependency Graph for Static Analysis of
Timed Rebeca Models 
}}\label{sec:rdg}

To compare two Timed Rebeca models, each model is first transformed into an
intermediate graph representation called a Rebeca Dependency Graph (RDG). The
semantic comparison can be reduced to graph traversal and comparing the nodes
and edges of the corresponding graphs, resulting in behavioral equivalence
modulo weak timed bisimilarity. The proposed RDG is a simplified version of the
RDG defined in~\cite{SabouriSirjani2010Actor}, primarily introduced for slicing
and reduction. In contrast, the RDG developed here is intended to support the
comparison of the observable behaviour of the BCTTSs derived from two Timed
Rebeca models. This connection is made precise by the soundness theorem below,
which shows that equivalent observable RDG slices imply weak timed
bisimulation, as defined in Definition~\ref{def:twt-equivalence}.


\subsection{Rebeca Dependence Graph}\label{subsec:rdg-overview}

The RDG is centered around six kinds of graph nodes: class-object nodes, entry points of execution, statements, formal-in/actual-in parameters, and state
variables. Using the edges, the graph makes explicit which class owns a
member, which entry activates a computation, which values flow into
expressions and assignments, and which send statements activate which message
servers. 
Accordingly,
the RDG is defined as follows.

\begin{definition}[Rebeca Dependence Graph]
A Rebeca dependence graph (RDG) for a Timed Rebeca model $\mathcal{M}=(C,\Ins,\conf)$ is a directed graph $\RDG(\mathcal{M})=(N,E)$ where $N$ and $E$ are the union of disjoint sets $N=N_S\cup N_C\cup N_E\cup N_{FP} \cup N_{AP}  \cup N_{SV}$, and $E=E_C\cup E_D \cup E_I \cup E_M\cup E_B\cup E_A $: 
\begin{itemize}
    \item $N_C$ denotes the class-object node set, one for each reactive class $c_i\in C$.
    
    \item $N_E$ denotes the entry-node set representing the entry points for executable elements. One entry node is considered for the
    \texttt{main} block, for each constructor, and for each message server.
    
    \item $N_{FP}$ denotes the formal-in parameter node set.  A formal-in node is considered for each formal parameter of a message server, representing an assignment statement that copies values from the message to the formal parameter.

    \item $N_{AP}$ denotes the actual-in parameter node set. An actual-in node represents an assignment statement copying the value of an actual parameter to the message.

    \item $N_{SV}$ denotes the state-variable node set, one for each state variable declaration of a reactive class.

    \item $N_S$ denotes the statement-node set. Each node is associated with an statement: assignments, nondeterministic assignments, conditional statements, send statements, and object-instantiation statements in \texttt{main}. 

    \item $E_C$ denotes the control-dependence edge set, connecting an entry node to its statements, i.e.,  each message server/constructor entry node to its formal-in parameters and statement body, and the \texttt{main} entry node to the object-instantiation statements in the \texttt{main} block. It also connects a conditional statement to its corresponding branches, each send statement to its actual parameters, and each object-instantiation statement to the corresponding class-object node. 

    \item $E_D$ denotes the data-dependence edge set, as defined in Sec. \ref{subsec::PDG}.
    
    \item $E_I$ denotes the intra-rebec dependence edge set. This dependency exists between the last statement of a message server which is assigning a value to a variable and the first use of that variable in another message server.
    
    \item $E_M$ denotes the member-dependence edge set, connecting class-object nodes to the members of that class, namely the state-variable nodes and the entry nodes of constructors, and message servers.
    
    \item $E_B$ denotes the binding-dependence edge set, connecting the actual-in parameter to its corresponding formal-in parameter. 
    
    \item $E_A$ denotes the activation edge set, connecting a send statement node to the entry node of the message server that is activated by that send.
    
\end{itemize}
We assume that each node has a content, denoted by $\cont(n)$, representing its corresponding syntactic content: 
\[
\begin{aligned}
&\forall C_i\in C\mathbin{\cdot} C_i=(c_i,M_i,V_i,K_i)
  \Rightarrow \exists n_c\in N_C\mathbin{\cdot} \cont(n_c)=c_i \wedge \\
&\quad \forall (v_i,T_i)\in V_i\mathbin{\cdot}\exists n_{v_i}\in N_{SV}\mathbin{\cdot}
  \cont(n_{v_i})=(v_i,T_i) \wedge (n_c,n_{v_i})\in E_M \wedge \\
&\quad \exists n_{cc}\in N_E \mathbin{\cdot} \cont(n_{cc})=\texttt{constructor}
  \wedge (n_c,n_{cc})\in E_M \wedge \\
&\quad \forall \mathit{msgsrv}\in M_i\mathbin{\cdot} \mathit{msgsrv}=(m,\nu,sts)
  \Rightarrow \exists n_m\in N_E\mathbin{\cdot} \cont(n_m)=m
  \wedge (n_c,n_m)\in E_M \wedge \\
&\qquad \forall (v,T)\in\nu\mathbin{\cdot} \exists n_p\in N_{FP}\mathbin{\cdot}
  \cont(n_p)=(v,T) \wedge (n_m,n_p)\in E_C \wedge \\
&\qquad \forall st\in sts\mathbin{\cdot} \exists n_s\in N_S\mathbin{\cdot}
  \cont(n_s)=st \wedge (n_m,n_s)\in E_C \wedge \\
&\quad \exists!\, n_{\mathit{main}}\in N_E\mathbin{\cdot}
  \cont(n_{\mathit{main}})=\texttt{main} \wedge
  \forall (x,c_i)\in\Ins\mathbin{\cdot} \exists n_I\in N_S\mathbin{\cdot}
  (n_{\mathit{main}},n_I)\in E_C \wedge (n_I,n_c)\in E_C
\end{aligned}
\]
\end{definition}
We use dot notation to access the elements of an RDG $\mathcal{G}$. For
example, $\mathcal{G}.N_{SV}$ and $\mathcal{G}.N_P$ denote its sets of state
variables and parameters, respectively. We write $n\overto{} n'$ when $(n,n')\in \mathcal{G}.E$.  

The content of an entry node is its qualified name, which is $\it main$ for the entry node of main block, and the message server's name for the message server's entry nodes. 
The content of a statement node is its statement form together with the
identifiers, expressions, receiver, target message server, and timing
annotations that occur in it. The content of a state variable or parameter node is its name and type. 

We remark that the data-dependence edges are defined within a message server. To follow-up the effect of statements on the state variables among the message servers, we consider the intra-rebec and binding dependence edges. Since message passing in Rebeca is asynchronous, execution does not move
directly to the called message server when the send statement is executed.
Instead, the send statement creates a message for the receiver, and the
activation edge records that this message may later enable execution of the
target message server. Therefore, we treat send statements as similar to other statements together with a set of assignments for actual parameters.  The main role of  member-dependence edges is to preserve ownership information and to ensure that if a member is relevant to a semantic comparison, the corresponding  reactive class is also reachable in the graph.


Our RDG definition is simpler than the one defined in~\cite{SabouriSirjani2010Actor} due to ignoring known rebecs, and treating send statements as normal statements. Originally a separate set of activation nodes together with activation edges capture send statements. In addition, constructors are
treated as ordinary entry nodes. So,  initialization remains within the same structural framework as the rest of the model.

\begin{example}
Figure~\ref{fig:rebec1} shows a simple Timed Rebeca model and its RDG.
The model contains two reactive classes, $A$ and $B$. Class $A$ sends
$\mathit{m1}$ to itself, and $\mathit{m1}$ sends $\mathit{m2}$ to $B$ depending on the value of its parameter $z$. Message
server $\mathit{m2}$ updates $x$ and $y$, and then either sends $\mathit{m2}$ to
itself or $\mathit{m1}$ to $A$ depending on the value of the state variable $y$.

The statement $y=x$ defines $y$, and $\mathit{if}(y)$ uses it. Therefore, there is a data-dependence edge among them. The statement $x=\neg x$ defines $x$, and its value is used when sending the message $m1$ within the parameter. Therefore, there is data-dependence edge from $x=\neg x$ to the actual-in parameter $z_{in}=x||y$. The actual-in parameter is connected to the formal-in parameter $z_{in}=z$. There is an intra-rebec dependency edge from the last assignment to the state variable $x$ in the constructor of class $B$ to its first usage in the message server $m2$. 

\end{example}

\begin{figure}
	\centering
\begin{subfigure}[b]{0.3\textwidth}
\begin{lstlisting}[language=Hrebeca, basicstyle=\small]
reactiveclass A(10){
  knownrebecs{ B b; }
  A(){ self.m1(true); }
  msgsrv m1(boolean z){  
    if (z) b.m2(); }}
reactiveclass B(10){
  knownrebecs{ A a; }
  B(){ x = true; }
  statevars {
    boolean x, y;}
  msgsrv m2(){
    y = x;
    x = !x;
    if (y)
      self.m2() after (5);
    else
      a.m1(x || y);}}
main {
	A a(b):();
	B b(a):();}
\end{lstlisting}
	\caption{A simple Timed Rebeca model}
	\label{fig:code1}
    \end{subfigure}
\hfill
\begin{subfigure}[b]{0.65\textwidth}
	\includegraphics[width=\textwidth]{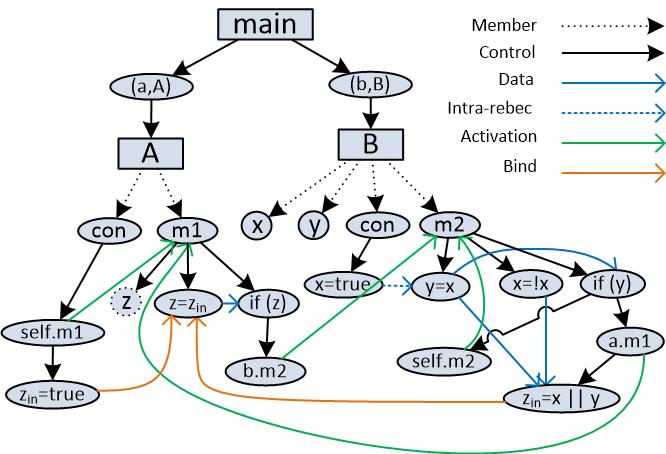}
    \caption{The Corresponding RDG: the nodes in the observable RDG slice with respect to $m1$ have been colored }
\label{Fig::RDG1}
\end{subfigure}
\caption{An example of a Timed Rebeca model and its RDG extraction}
\label{fig:rebec1}
\end{figure}

\subsection{Slicing RDGs and the Equivalence Relation on RDGs}\label{subsec::slice}
The comparison is performed on the parts of the models that are relevant to
the observable messages. Let $O$ be the set of observable messages given as
input.  The following definition makes this part of an RDG precise. We write $\node(\mathcal{G},m)$ to identify the node in a RDG $\mathcal{G}$ that correspond to the entry nodes of the message server $m$, i.e, $\node(\mathcal{G},m)= n$ where $n\in \mathcal{G}.N_E\wedge \cont(n)=m$. We trivially promote it for the set of message names $\node(\mathcal{G},O)$.

\begin{definition}[Observable RDG Slice]
\label{def:observable-rdg-slice}
Let $\mathcal{G}=(N,E)$ be the RDG of a Timed Rebeca model $\mathcal{M}$. 
The Slice of $\mathcal{G}$ with respect to the set of observable message names $O$, denoted by $\Slice(\mathcal{G},O)=(N',E')$, is a subgraph of $\mathcal{G}$, where \[
\begin{array}{l}
N' = \{\,n\in N\,\mid\,
\exists n_m\in \node(\mathcal{G},O)\cdot  n \overto{}^* n_m\,\}\,\cup \,\{ v \in N_{SV} \cup N_{FP}\,\mid\, \exists st\in N_{S} \text{use or define $v$}\}\\
E' = \{\,(n_1,n_2)\in E\,\mid\,
n_1\in N' \wedge n_2\in N'\}
\end{array}
\]
\end{definition}

The slice is therefore obtained by backward reachability from the observable
message-server entries. It contains the statements, actual-in parameters, activation sites, class-object nodes, and dependency edges that can affect an observable message directly or through a chain of RDG dependencies. In particular, every message-server entry retained in the slice has a directed
dependency path to an observable message-server entry. The messages represented
by these entry nodes are called relevant messages in
Definition~\ref{def:relevant-messages}. The slice also includes the state-variable and parameter declarations used or
defined by the reachable statements.

\begin{example}
Figure~\ref{fig:rebec2} shows a modified version of the simple Timed Rebeca model given in Figure \ref{fig:rebec1} and its RDG.
Class $A$ sends $\mathit{m1}$ to itself, and $\mathit{m1}$ sends $\mathit{m2}$ to $B$ with no restriction. Message server $\mathit{m2}$  has no change.

We call the Timed Rebeca model in Figure~\ref{fig:code1} and it revision in Figure~\ref{fig:code2}, $\mathcal{M}_{\it AB}$ and $\mathcal{M}_{\it AB}'$, respectively.  We computed the $\Slice(\mathcal{M}_{\it AB},\{m1\})$ and $\Slice(\mathcal{M}_{\it AB}',\{m1\})$ by coloring the nodes in Figures \ref{Fig::RDG1} and \ref{Fig::RDG2}, respectively. As $m1$ sends message $m2$ irrespective to the value of its parameter $z$, the nodes $z_{in}=x||y$ and $x=!x$ in the message server $m2$, $z=z_{in}$ in the message server $m1$, and $z_{in}=true$ in the constructor $A$ are not reachable. Therefore, the parameter $z$ is useless in this revision. 

\end{example}

\begin{figure}
	\centering
\begin{subfigure}[b]{0.3\textwidth}
\begin{lstlisting}[language=Hrebeca, basicstyle=\small]
reactiveclass A(10){
  knownrebecs{ B b; }
  A(){ self.m1(true); }
  msgsrv m1(boolean z){  
    b.m2(); }}
reactiveclass B(10){
  knownrebecs{ A a; }
  B(){ x = true; }
  statevars {
    boolean x, y;}
  msgsrv m2(){
    y = x;
    x = !x;
    if (y)
      self.m2() after (5);
    else
      a.m1(x || y);}}
main {
	A a(b):();
	B b(a):();}
\end{lstlisting}
	\caption{A modified version of the simple Timed Rebeca model}
	\label{fig:code2}
    \end{subfigure}
\hfill
\begin{subfigure}[b]{0.65\textwidth}
	\includegraphics[width=\textwidth]{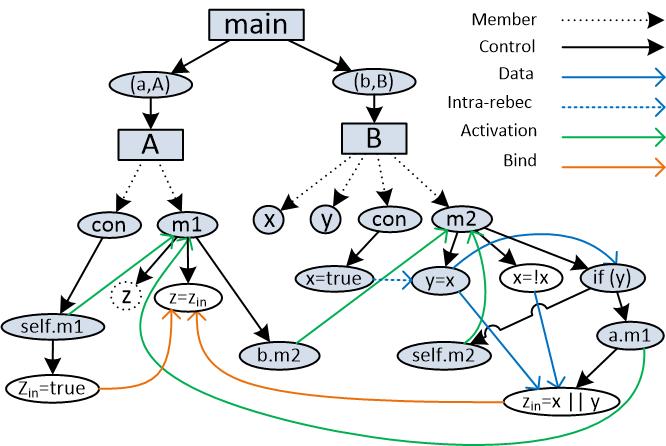}
    \caption{The Corresponding RDG: the nodes in the observable RDG slice with respect to $m1$ have been colored }
\label{Fig::RDG2}
\end{subfigure}
\caption{An example of a Timed Rebeca model and its RDG extraction}
\label{fig:rebec2}
\end{figure}

Intuitively, two RDGs are called equivalent if they are isomorphic; preserving the same structure among the same kind of nodes with the same content. 

\begin{definition}[RDG Equivalence]
\label{def:rdg-slice-equivalence}
Given two RDGs $G_1=(N^1,E^1)$ and $G_2=(N^2,E^2)$ are called equivalent, denoted by $G_1\equiv G_2$, if and only if there is a bijective mapping $\phi: N^1\rightarrow N^2$ such that \begin{itemize}
\item $
\forall n\in N_1\cdot n\in G_1.N_X \Leftrightarrow \phi(n)\in G_2.N_X \wedge 
\cont(n)=\cont(n')$, where $X\in\{C,E,P,{\it SV},S\}$;
\item $\forall n\in {N^1_C},v\in N^1_{SV}\cdot (n,v)\in E_M^1 \Leftrightarrow (\phi(n),\phi(v))\in E_M^2$;
\item $\forall n\in {N^1_C},e\in N^1_{E}\cdot (n,e)\in E_M^1 \Leftrightarrow (\phi(n),\phi(e))\in E_M^2$;
\item $\forall n\in {N^1_E},p\in N^1_{FP}\cdot  (n,p)\in E_C^1 \Leftrightarrow (\phi(n),\phi(p))\in E_C^2$;
\item $\forall e\in N^1_{E}, st\in N_S^1\cdot (e,st)\in E_C^1 \Leftrightarrow (\phi(e),\phi(st))\in E_C^2$;
\item $\forall st\in N_S^1 , e\in N^1_{E}\cdot (st,e)\in E_A^1 \Leftrightarrow (\phi(st),\phi(e))\in E_A^2$;
\item $\forall st_1\in N_S^1 , st_2\in N_S^1\cdot(st_1,st_2)\in E_C^1 \Leftrightarrow (\phi(st_1),d,\phi(st_2))\in E_C^2$;
\item $\forall st_1\in N_S^1 , st_2\in N_S^1\cdot(st_1,st_2)\in E_D^1 \Leftrightarrow (\phi(st_1),\phi(st_2))\in E_D^2$;
\item $\forall st_1\in N_S^1, p'\in N_{AP}^1\cdot (st_1,p')\in E_C^1 \Leftrightarrow (\phi(st_1),\phi(p'))\in E_C^2$;

\end{itemize}
\end{definition}

\subsection{Static Comparison of Timed Rebeca Models}

A variable is relevant to a slice of an RDG if its corresponding state-variable node or parameter node belongs to the observable RDG slice. These variables are those whose values may affect the guards, assignments, message activations, timing annotations, or parameter values that contribute to the observable behaviour.

\begin{definition}[Relevant Variables] Given the Timed Rebeca model $\mathcal{M}$ and the set of observable message names $O$, the set of relevant variables are those reachable state variables and parameters in $\Slice(\RDG(\mathcal{M}),O)$ that affect its observational behavior: 
\[RV(\mathcal{M},O)= \Slice(\RDG(\mathcal{M}),O).N_{SV} \cup \Slice(\RDG(\mathcal{M}),O).N_{FP}.\]
\end{definition}

A message is relevant if its message-server entry node belongs to the observable
RDG slice. Hence, the relevant messages include the observable messages together with the internal messages that are
reached from the observable messages in the slice because they may affect the production, scheduling, or parameter
values of observable messages.

\begin{definition}[Relevant Message Names]\label{def:relevant-messages} Given the Timed Rebeca model $\mathcal{M}$ and the set of observable message names $O$, the set of relevant message names are those messages that affect its observational behavior, defined as:
\[
RM(\mathcal{M},O)=\{{\it cont}(n) \,\mid\, n \in \Slice(\RDG(\mathcal{M}),O).N_E \wedge \cont(n)\in \Name\}
\]
\end{definition}



The Proposition \ref{Pro::stutter} explains that the state changes due to the handling of the irrelevant message names are shift-equivalent with respect to the relevant message names and variables. 
\begin{proposition}\label{Pro::stutter}
Given the Timed Rebeca model $\mathcal{M}$ and the set of observable message names $O$, let $\mathcal{A}_{RM}(\CTTS(\mathcal{M}))=(S,\rightarrow,s_0)$, where $RV=RV(\mathcal{M},O)$ and $RM=RM(\mathcal{M},O)$. For any $s\in S$ if  $s\Rightarrow s'$, then $s\simeq_{RV}^{RM}s'$.
\end{proposition}

\begin{proof}
When $s\overto{\tau}s'$, it is due to the handling of an irrelevant message $\m=(y,x,m,par,t)$ by the actor $x$ such that  $m\not\in RM$. Def. \ref{def:observable-rdg-slice} and \ref{def:relevant-messages} imply that the entry node corresponding to the message server $m$ in the RDG $\Slice(\RDG(\mathcal{M}),O)$ does not reach any message server of message names in $RM$. In other words, its process does not trigger any relevant messages or affect the relevant variables. Assuming that $s=(\mathfrak{s},t)$ and $\mathfrak{s}(x)=(\valuation,\bag,\epsilon)$, by application of $\textsc{TakeMessage}$, it holds that $(\valuation\uplus par,\body(x,m),\epsilon),t\aoverto{}^*(\valuation',\epsilon,out),t$ where $out\downharpoonright RM = \epsilon$ and $\valuation \downharpoonright RV = \valuation' \downharpoonright RV$. Therefore, upon dispatching the massages of $out$, i.e., $s'=(\upd(\mathfrak{s}[x\mapsto (\valuation'\downharpoonright \vars(x),\bag\setminus \m,\epsilon)],out),t)$, the bag of rebecs in the resulting state $s'$ is not expanded by any relevant message. Taking $\delta=0$, by Def. \ref{Def::shift}, it results in $s\simeq_{RV}^{RM}s'$. The same holds for a sequence of $\tau$-transitions.
\end{proof}

The following theorem indicates that two Timed Rebeca models with equivalent observable slices with respect to the message names $O$, induce weakly timed bisimilar semantic models. 

\begin{theorem}[Soundness of the Syntax-level Comparison]
Given the timed Rebeca models $\mathcal{M}_1$ and $\mathcal{M}_2$, and a set of observable message names $O$, if
the observable RDG slices of the two models are equivalent, then the BCTTSs
generated from the two models are weak timed bisimilar:
\[
\Slice(\RDG(\mathcal{M}_1),O)\equiv \Slice(\RDG(\mathcal{M}_2),O) \Rightarrow \mathcal{A}_O(\BCTTS(\mathcal{M}_1))\simeq_{wtb} \mathcal{A}_O(\BCTTS(\mathcal{M}_2))
\]
\end{theorem}
\begin{proof}
The assumption $\Slice(\RDG(\mathcal{M}_1),O)\equiv \Slice(\RDG(\mathcal{M}_2),O)$ 
implies that $
RM=RM(\mathcal{M}_1,O)=RM(\mathcal{M}_2,O)$ and $RV=RV(\mathcal{M}_1,O)=RV(\mathcal{M}_2,O)$. We first prove that $\mathcal{A}_{\it RM}(\BCTTS(\mathcal{M}_1))\simeq_{wtb} \mathcal{A}_{\it RM}(\BCTTS(\mathcal{M}_2))$, where $O\subseteq RM$. Then, the proof of the theorem immediately results from Proposition \ref{Pro::abstraction}. We construct the relation $\mathcal{R}$ and show that it satisfies the transfer conditions of weak timed bisimulation given in Def. \ref{def:twt-equivalence}:
\[
\mathcal{R}=\{(s,u)\,\mid\, s\simeq_{RM}^{RV}u \}.
\]

For an arbitrary pair $s\mathcal{R} u$, two cases can be considered:\begin{itemize}
\item $s\xRightarrow{m}s'$: there exists $s_1$ and $s_2$ such that $s\Rightarrow s_1 \overto{m}s_2 \Rightarrow s$. 
By construction of $\mathcal{R}$ and Proposition \ref{Pro::stutter}, it holds that $s_1 \mathcal{R} u$. The transition $s_1\overto{m}s_2$ has derived by application of $\textsc{TakeMessage}$. Thus, there exists a rebec $x$ such that $\map(s_1)(x)=(\valuation, \bag,\epsilon)$ and $(y,x,m,par,\now(s_1))\in \bag$. Assume that $\map(u)(x)=(\valuation',\bag',\epsilon)$. As $s_1\simeq_{RM}^{RV}u $, we conclude that there exists $\delta\in\int$ such that $\now(s_1)=\now(u)+\delta$, $\valuation \downharpoonright RV =\valuation '\downharpoonright RV$, and $(y,x,m,par,\now(s_1)+\delta)\in\bag'\downharpoonright RM$. As $m\in RM$, it holds that $(y,x,m,par,\now(u))\in\bag'$ and the rule $\textsc{TakeMessage}$ can be applied, leading to $u\overto{m}u'$. We should prove that \[
\begin{array}{l}
((\valuation,st,\epsilon),\now(s_1)\aoverto{}^*(\valuation_1,\epsilon,out_1),\now(s_1) \wedge
(\valuation',st',\epsilon),\now(u)\aoverto{}^*(\valuation_2,\epsilon,out_2),\now(u) )\Rightarrow \\
\hspace*{2cm}(\valuation_1 \downharpoonright RV = \valuation_2  \downharpoonright RV \,\wedge \, \forall (x,z,m',p,a')\in out_1 \downharpoonright RM \Leftrightarrow (x,z,m',p,a'+\delta)\in out_2 \downharpoonright RM).
\end{array}\]where $st=\body(x,m)$ with respect to $\mathcal{M}_1$ and $st'=\body(x,m)$ with respect to $\mathcal{M}_2$. This trivially is resulted from the equality of the observable slices and Theorem \ref{The::StrongEquvalence}. As the bag of all rebecs are updated with same relevant messages and the local variables of rebec $x$ with respect to $RV$ has been updated similarly, we conclude that $s_2 \simeq_{RM}^{RV} u'$. The application of Proposition \ref{Pro::stutter} implies that $s'\mathcal{R} u'$.
\item $s\toverto{t}s'$ where there exists at least one relevant message $(y,x,m,par,\arv)$ in the bag of rebec $x$, where $m\in RM$ and $\arv-\now(s)=t$, that can be triggered in $s'$, i.e., $\now(s')=\arv$: there exists $s_1,\ldots,s_{n-1}$ such that $s\toverto{t_1} s_1 \toverto{t_2} s_2 \ldots s_{n-1}\toverto{t_n} s$, where $t=t_1+\ldots+t_n$. It is obvious that $(y,x,m,par,\arv)$ is contained in the bag of rebec $x$ in the state $s$. If it is not contained, the processing of the irrelevant messages over the path $s\toverto{t}s'$ would have triggered $m$ which is in contradiction with definition of relevant messages. The assumption $s\simeq_{RM}^{RV}u$ implies that $(y,x,m,par,\arv+\delta)$ is also contained in the bag of rebec $x$, where $\now(u)=\now(s)+\delta$ and $\arv+\delta-\now(u)=t$. Therefore, there exists $u_1,\ldots,u_{m-1}$ such that $u\toverto{t_1'} u_1 \toverto{t_2'} u_2 \ldots s_{n-1}\toverto{t_m'} u'$, where $t=t_1'+\ldots+t_n'$, and $\now(u')=\now(u)+t$. Thus, $u\toverto{t}u'$ and $s\simeq_{RM}^{RV}u$. 
\end{itemize}
The same discussion can be managed when $u\xRightarrow{m}u'$ or $u\toverto{t}u'$. Thus, $\mathcal{R}$ is a weak timed bisimulation. 
\end{proof}

\subsection{Discussion}
The RDG-based comparison supports the following kinds of model revisions:

\begin{itemize}[
    leftmargin=1.5em,
    itemsep=1pt,
    topsep=2pt,
    parsep=0pt,
    partopsep=0pt
]
    \item Adding or removing state variables whose nodes do not belong to the
    observable RDG slice. Such variables do not affect the observable
    behavior.

    \item Adding or removing non-observable messages whose message-server
    entries have no directed dependency path to an observable message-server
    entry. A non-observable message that may affect an observable message is
    relevant and must therefore be preserved.

    \item Adding or removing a rebec whose class, state variables, message
    servers, and dependencies are entirely outside the observable RDG slice.

    \item Adding or removing message parameters when the corresponding actual-in parameter nodes are outside the observable RDG
    slice. 

    \item Adding or removing statements to a message body when the corresponding statement nodes are outside the observable RDG slice.
\end{itemize}

\section{Case study on Smart Home}\label{sec:case}
Figure~\ref{fig:heater-models} presents two Timed Rebeca models of a smart home consisting of four kinds of rebecs: controller, heating unit, room, and sensor. The original model $M$ contains the
behaviour relevant to the heater-control loop: the room updates its temperature,
the sensor sends the measured value to the controller, and the controller
activates or switches off the heater according to the received temperature.
The revised model $M'$ preserves this behaviour but adds a \emph{Notifyer}
rebec and additional calls to \emph{notify}. We choose the observable set $O$
to contain only the messages that belong to the heater-control behaviour.
Therefore, the notification part is outside the selected observable slice.

\begin{figure}
    \centering
    \setlength{\tabcolsep}{5pt}

    \begin{tabular}{@{}p{0.49\textwidth}@{\hfill}p{0.49\textwidth}@{}}

        \begin{minipage}[t]{\linewidth}
            \vspace{0pt}

\begin{lstlisting}[
                language=Hrebeca,
                numbers=left,
                frame=lines,
                framerule=0.3pt,
                xleftmargin=2.2em,
                numbersep=6pt,
                basicstyle={\fontsize{6.4}{7.9}\selectfont\ttfamily},
                numberstyle={\fontsize{5.0}{6.0}\selectfont},
                columns=fullflexible,
                keepspaces=true,
                breaklines=true,
                breakatwhitespace=true,
                showstringspaces=false,
                aboveskip=0pt,
                belowskip=0pt
            ]
reactiveclass Controller(5) {
	knownrebecs{ HC_Unit hc_unit; }
	statevars{ boolean heating_active; int sensedValue; }
	Controller(){
		heating_active = false;
		sensedValue = 20; }
	msgsrv getSense(int temp) {
		sensedValue = temp;
		if (21 > temp && heating_active == false) {
			hc_unit.activateh(); //heat
			heating_active = true;
		} else if (21 <= temp && heating_active == true) {
			hc_unit.switchoff();
			heating_active = false;
		} }
}

reactiveclass Room(5) {
	knownrebecs{ Sensor sensor; }
	statevars{ int temperature, outside_air_blowing, regulation; }
	Room(){
		temperature = 21;
		regulation = 0;
		outside_air_blowing = 0;
		self.tempchange(); }
	msgsrv tempchange() {
		outside_air_blowing = ?(1,0);
		temperature = temperature - outside_air_blowing + regulation;
		sensor.getTemp(temperature);
		self.tempchange() after(10); }
	msgsrv regulate(int v) {
		regulation = v;
	}
}

reactiveclass Sensor(5) {
	knownrebecs{ Room room; Controller controller; }
	msgsrv getTemp(int temp) {
		controller.getSense(temp);
	}
}

reactiveclass HC_Unit(5) {
	knownrebecs{ Room room; }
	statevars{ boolean heater_on; }
	HC_Unit(){
		heater_on = false;
	}
	msgsrv activateh() {
		room.regulate(1);
		heater_on = true; }
	msgsrv switchoff() {
		room.regulate(0);
		heater_on = false; }
}

main{
	Room room(sensor):();
	Controller controller(hc_unit):();
	Sensor sensor(room,controller):();
	HC_Unit hc_unit(room):();
}
\end{lstlisting}

            \vspace{0.35em}

            \centering\small
            (a) Original heater model

        \end{minipage}

        &

        \begin{minipage}[t]{\linewidth}
            \vspace{0pt}

\begin{lstlisting}[
                language=Hrebeca,
                numbers=left,
                frame=lines,
                framerule=0.3pt,
                xleftmargin=2.2em,
                numbersep=6pt,
                basicstyle={\fontsize{6.4}{7.9}\selectfont\ttfamily},
                numberstyle={\fontsize{5.0}{6.0}\selectfont},
                columns=fullflexible,
                keepspaces=true,
                breaklines=true,
                breakatwhitespace=true,
                showstringspaces=false,
                aboveskip=0pt,
                belowskip=0pt
            ]
reactiveclass Controller(5) {
	knownrebecs{ HC_Unit hc_unit; Notifyer notifyer; }
	statevars{ boolean heating_active; int sensedValue; }
	Controller(){
		heating_active = false;
		sensedValue = 20; }
	msgsrv getSense(int temp) {
		sensedValue = temp;
		if (21 > temp && heating_active == false) {
			hc_unit.activateh(); //heat
			heating_active = true;
			notifyer.notify(temp);
		} else if (21 <= temp && heating_active == true) {
			hc_unit.switchoff();
			heating_active = false;
			notifyer.notify(temp); } }
}

reactiveclass Room(5) {
	knownrebecs{ Sensor sensor; }
	statevars{ int temperature, outside_air_blowing, regulation; }
	Room(){
		temperature = 21;
		regulation = 0;
		outside_air_blowing = 0;
		self.tempchange(); }
	msgsrv tempchange() {
		outside_air_blowing = ?(1,0);
		temperature = temperature - outside_air_blowing + regulation;
		sensor.getTemp(temperature);
		self.tempchange() after(10); }
	msgsrv regulate(int v) { 
		regulation = v; 
	}
}

reactiveclass Sensor(5) {
	knownrebecs{ Room room; Controller controller; }
	msgsrv getTemp(int temp) { 
		controller.getSense(temp); 
	}
}

reactiveclass HC_Unit(5) {
	knownrebecs{ Room room; }
	statevars{ boolean heater_on; }
	HC_Unit(){ 
		heater_on = false; 
	}
	msgsrv activateh() {
		room.regulate(1);
		heater_on = true; }
	msgsrv switchoff() {
		room.regulate(0);
		heater_on = false; }
}

reactiveclass Notifyer(5) {
	statevars{ int current_status; }
	Notifyer(){ }
	msgsrv notify(int current_status_) {
		current_status = current_status_;
		send_signal() after(3); }
	msgsrv send_signal(){ }
}

main{
	Room room(sensor):();
	Controller controller(hc_unit,notifyer):();
	Sensor sensor(room,controller):();
	HC_Unit hc_unit(room):();
	Notifyer notifyer():();
}
\end{lstlisting}

            \vspace{0.35em}

            \centering\small
            (b) Heater model extended with a \texttt{Notifyer}

        \end{minipage}

    \end{tabular}

    \caption{
        Two heater-controller models used in the RDG-based comparison:
        (a) the original model $M$, and
        (b) the extended model $M'$, which adds a \texttt{Notifyer} rebec
        while preserving the heater-control behaviour.
    }

    \label{fig:heater-models}
\end{figure}

Figure~
\ref{fig:heater-rdg-extended} show the complete RDGs of the model $M'$
together with the observable slices induced by $O$ (blue nodes). Since the added
\texttt{Notifyer} behaviour (yellow nodes) does not affect the selected observable messages, its nodes remain outside the observable slice of $M'$. 






\begin{figure}[H]
    \centering

    \includegraphics[
        width=0.92\textwidth,
        keepaspectratio
    ]{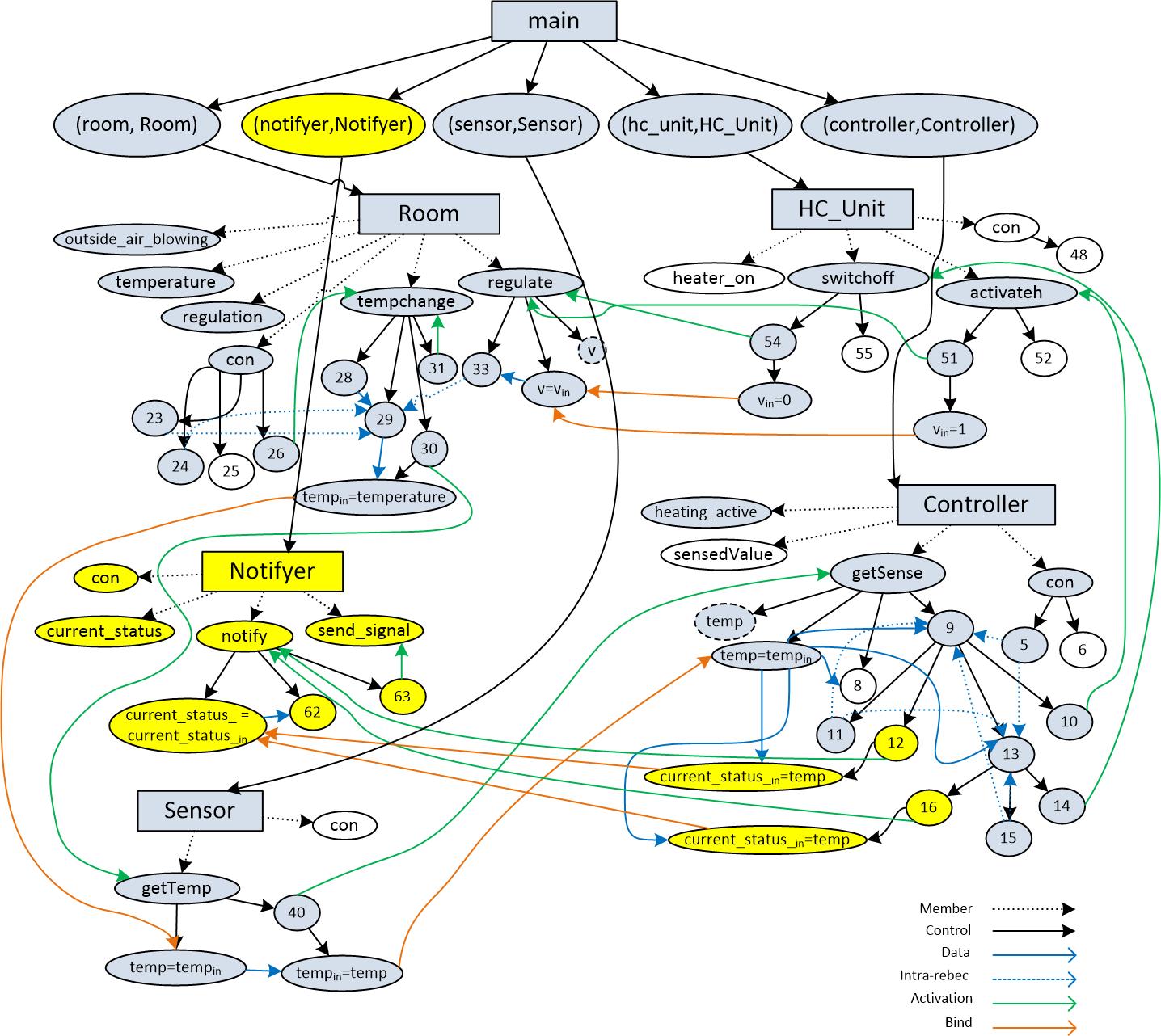}

    \captionsetup{
        format=plain,
        justification=justified,
        singlelinecheck=false
    }

    \caption{
    RDG of the revised heater model $M'$. Blue nodes indicate the observable
    slice induced by $O$, white nodes lie outside the slices of both models,
    and yellow nodes were added by the revision. 
    }

    \label{fig:heater-rdg-extended}
\end{figure}

The nodes belonging to the \texttt{Notifyer} reactive class do not have a
directed dependency path to an entry node of a message server in $O$.
Therefore, these nodes cannot affect an observable message and do not belong
to the slice induced by $O$.

\section{Experiments}\label{sec:experiments}

We compare the cost of checking observable RDG slices with the cost of
state-space exploration and tiny-twin reduction on ten models. We address the
following research questions:
\begin{description}[style=unboxed,leftmargin=0pt]
\item[RQ1 (Cost).] \emph{How do the execution time and peak memory of the
static comparison differ from those of tiny-twin generation by state-space
exploration and reduction, and how do these costs relate to the sizes of the
source models and of their state spaces?} Reusing an existing tiny twin is
worthwhile only if checking a revision costs substantially less than
regenerating the twin. Section~\ref{subsec:eval-results} answers RQ1 using the
measurements in Table~\ref{tab:eval-comparison}.
\item[RQ2 (Slice size).] \emph{What fraction of each model's RDG does its
observable slice retain, and which dependency structures keep this fraction
small?} Revisions confined to the part of a model outside its observable slice
do not affect the comparison, so a smaller slice leaves more of the model free
to evolve while reuse can still be established.
Table~\ref{tab:eval-comparison} reports these fractions, and
Section~\ref{subsec:eval-slice} relates them to the dependency structures of
the models.
\item[RQ3 (Change attribution).] \emph{When a revision changes a message server
shared by several message-dependency loops, to what extent do the static
verdicts obtained with each loop's observable messages agree with changes in
that loop's regenerated tiny twin?} Soundness guarantees reuse only for
accepted comparisons; a rejected comparison carries no such guarantee, and a
component diagram shows which loops exist but not which of them a revision
affects. Section~\ref{sec:discussion} answers RQ3 by comparing the verdicts for
three bounded-buffer revisions with their regenerated twins.
\end{description}

\subsection{Cost of the two computations}\label{subsec:eval-costs}

For a semantic model $\BCTTS(\mathcal{M})=(S,\rightarrow,s_0)$, explicit
state-space exploration visits the reachable states and transitions. Traversing
this graph takes $\Theta(|S|+|\rightarrow|)$ time, excluding the cost of computing
and storing each successor. Reduction also operates on the generated structure.
The number of states can grow combinatorially with message orders and
interleavings, independently of the number of source statements.

Slice-equivalence checking works on code. For an RDG $G=(N,E)$, the backward
traversal that computes $\Slice(G,O)$ takes $\mathcal{O}(|N|+|E|)$ time.
The full decision also includes parsing, RDG construction, and checking slice
isomorphism. Its cost depends on the size and structure of graphs representing
source statements and their dependencies.

\subsection{Tools}\label{subsec:eval-tools}

\emph{AdaptChecker}~\cite{AdaptChecker} takes an original and a revised Timed
Rebeca model and a set $O$ of observable messages. It constructs both RDGs and
computes their observable slices by following data, control, and message
dependencies backward from the message servers in $O$. It then checks for a
node bijection preserving node kinds, contents, and typed dependencies.
The comparison works directly on the source graphs, without generating either
model's state space. The tool produces a verdict, a report of the comparison,
and graph visualizations showing the observable slices.

\emph{TiniActor}~\cite{TiniActor} takes a generated state space in Afra's
output format and the observable set $O$. It casts the state space to an LTS
and replaces message actions outside $O$ by $\tau$. Its pipeline applies a
weak reduction, accumulates time
(Definitions~\ref{def:twt-equivalence} and~\ref{Def::weaktrace}), and reduces
again to obtain the tiny twin. It supports weak bisimulation and weak trace
reduction; the reported twins use weak trace reduction through
\texttt{ltsconvert} in mCRL2 202607.0~\cite{MCRL2Tools}. The output consists of
the reduced LTS and its graph representation for inspection and comparison.

\subsection{Benchmark models and measurement protocol}\label{subsec:eval-models}

Models 1--9 abstract control systems; model~10 is a concurrency benchmark.
Table~\ref{tab:eval-models} identifies their sources and observable-set sizes.
Models 1--6 are hand-written architectural abstractions inspired by the cited
systems. Models 7--10 adapt published Lingua Franca examples, preserving selected
message flows while replacing numerical computations by finite data domains.
Plant feedback and some request--response interactions are abstracted away.
The models we check have no \texttt{delay} statements or local variables.
The models, observable sets, and source material are available in the
artifact~\cite{TimedActorBenchmarks}.

\begin{table}[H]
\centering
\caption{Benchmark provenance. $|O|$ counts instance-qualified message names;
model~10 has separate producer ($O_P$) and consumer ($O_C$) criteria.}
\label{tab:eval-models}
\small
\begin{tabular*}{\textwidth}{@{\extracolsep{\fill}}cllr@{}}
\hline
\# & System & Architectural source or adapted example & $|O|$ \\
\hline
1 & Wind turbine controller & ROSCO~\cite{ROSCOSource} & 3 \\
2 & Automotive body control & LIN~\cite{LINStandard} & 2 \\
3 & CubeSat flight software & NASA F Prime~\cite{FPrimeSource} & 2 \\
4 & Digital substation & IEC 61850/libIEC61850~\cite{IEC61850Source} & 2 \\
5 & Integrated modular avionics & ARINC 653/XtratuM~\cite{XtratuMSource} & 2 \\
6 & Modular battery pack & foxBMS~\cite{FoxBMSSource} & 3 \\
7 & Longitudinal flight control & ROSACE in LF~\cite{PagettiEtAl2014ROSACE,LFROSACESource} & 3 \\
8 & Adaptive robot arm control & MuJoCo robot arm in LF~\cite{LFRobotArmSource} & 3 \\
9 & Autonomous valet parking & Autoware pipeline in LF~\cite{LFAVPSource} & 4 \\
10 & Bounded buffer & Savina bounded buffer in LF~\cite{LFSavinaSource} & 3 each \\
\hline
\end{tabular*}
\end{table}

The experiments use a shared server with two AMD EPYC 9554 processors. Our
environment exposes 112 cores and 224 hardware threads, 1{,}511\,GiB of RAM,
and no swap. State-space exploration and reduction were run once per original
model. Wall-clock times are indicative because the machine is shared.

For each static cost measurement, we increase one integer literal by seven in
a statement outside the observable slice and verify that comparison accepts
the revision. We report the median and interquartile range (IQR) of 20 fresh
Python 3.12.3 processes per pair. The decision timer covers observable-file
loading, parsing both models, building both RDGs, slicing, and comparison;
it excludes interpreter startup, imports, and report or graph-file output.
Memory is the maximum process peak resident set size across these runs.
For both tiny-twin generation and static comparison, time and memory use can
vary with the changes made to the model.

\subsection{Results}\label{subsec:eval-results}

Table~\ref{tab:eval-comparison} places the source sizes and both computations
side by side. Dynamic time is the sum of exploration and reduction times;
dynamic memory is the larger peak of those two stages. The tiny-twin counts
come from the stored LTS files. Source lines exclude comments and blank lines.

\begin{table}[H]
\centering
\caption{Source and semantic sizes, with tiny-twin generation and static
comparison costs. Slash pairs give counts in the order shown; K=$10^3$,
M=$10^6$. Dynamic memory is in GiB, static memory in MiB. A dash denotes an
unrecorded measurement.}
\label{tab:eval-comparison}
\label{tab:eval-dynamic}
\label{tab:eval-static}
\footnotesize
\setlength{\tabcolsep}{3pt}
\begin{tabular*}{\textwidth}{@{\extracolsep{\fill}}crrrrrrrrr@{}}
\hline
 & \multicolumn{3}{c}{Source and dependence graphs} &
 \multicolumn{4}{c}{Exploration + TiniActor} & \multicolumn{2}{c}{AdaptChecker} \\
\cline{2-4}\cline{5-8}\cline{9-10}
\# & Classes/ & RDG & Slice nodes & BCTTS & Twin & Total & Peak & Decision & Peak \\
 & lines & nodes/edges & (\% RDG) & states/trans. & states/trans. & time & GiB & ms (IQR) & MiB \\
\hline
1 & 17/469 & 327/548 & 61 (18.7) & 127K/326K & 13/16 & 37.1s & 7.4 & 18.2 (0.4) & 21.0 \\
2 & 27/637 & 417/721 & 43 (10.3) & 48.6K/85.1K & 12/15 & 17.9s & 3.6 & 20.9 (0.2) & 21.0 \\
3 & 32/650 & 430/694 & 48 (11.2) & 111K/301K & 16/24 & 54.3s & 10.0 & 21.6 (0.2) & 21.0 \\
4 & 35/704 & 473/712 & 50 (10.6) & 78.9K/138K & 10/14 & 41.6s & 7.9 & 28.1 (0.3) & 21.0 \\
5 & 34/732 & 515/811 & 28 (5.4) & 74.6K/105K & 12/18 & 36.8s & 7.5 & 23.4 (0.1) & 21.0 \\
6 & 10/214 & 148/231 & 55 (37.2) & 11.0M/55.3M & 29/41 & 79m35s & 538 & 10.1 (0.1) & 21.0 \\
7 & 19/388 & 277/433 & 60 (21.7) & 1.81M/3.00M & 89/141 & 10m23s & 110.3 & 15.3 (0.2) & 21.0 \\
8 & 12/291 & 221/363 & 118 (53.4) & 734K/2.40M & 30/62 & 3m27s & 32.8 & 27.6 (0.3) & 21.0 \\
9 & 20/395 & 282/456 & 187 (66.3) & 108K/307K & 43/61 & 37.9s & 6.78 & 42.2 (0.5) & 22.8 \\
$10_P$ & 4/112 & 87/150 & 44 (50.6) & 603/940 & 23/32 & 0.18s & 0.029 & 7.4 (0.1) & 21.0 \\
$10_C$ & 4/112 & 87/150 & 59 (67.8) & 603/940 & 20/23 & -- & -- & 8.5 (0.1) & 21.0 \\
\hline
\end{tabular*}
\end{table}

All static comparisons take less than 44\,ms and use less than 23\,MiB.
The nearly constant peak-memory values indicate that the baseline footprint of
the Python process dominates these measurements, so the incremental dependence
on RDG size is not visible at the resolution of process-level peak-RSS
measurements. The largest gap occurs for model~6, whose tiny-twin generation takes over an
hour and hundreds of gigabytes despite its small source. Its many reachable
states make exploration and reduction expensive, while the static check
operates on a small RDG.

\subsection{Models suited to slice-equivalence checking}
\label{subsec:eval-slice}

The method is most useful when $\Slice(G,O)$ contains a small part of the RDG.
This is favored by a wide fan-out before the observed branch and little feedback
from the end of the chain to its start. When the observable messages belong to
one branch, backward slicing removes the other branches and components that
cannot affect them~\cite{SabouriSirjani2010Actor}. Most of our samples follow
these structural patterns.

When the observable slice contains most RDG nodes, more source changes can
affect the slice, reducing the opportunities to establish reuse. A rejected
comparison cannot guarantee reuse, so the tiny twin must be regenerated or
checked behaviorally.

\section{Discussion}\label{sec:discussion}

\subsection{Distinguishing causality loops under evolution}
\label{subsec:eval-causality}

The bounded buffer contains two producers and two consumers. One message loop
connects grants, production, and deposits; the other connects consumer
readiness and delivery. These are cycles of asynchronous message dependencies,
not instantaneous causality cycles in Lingua Franca's scheduling semantics.
Figure~\ref{fig:eval-causality-diagram} shows both connections.

\begin{figure}[H]
\centering
\includegraphics[width=0.6\textwidth]{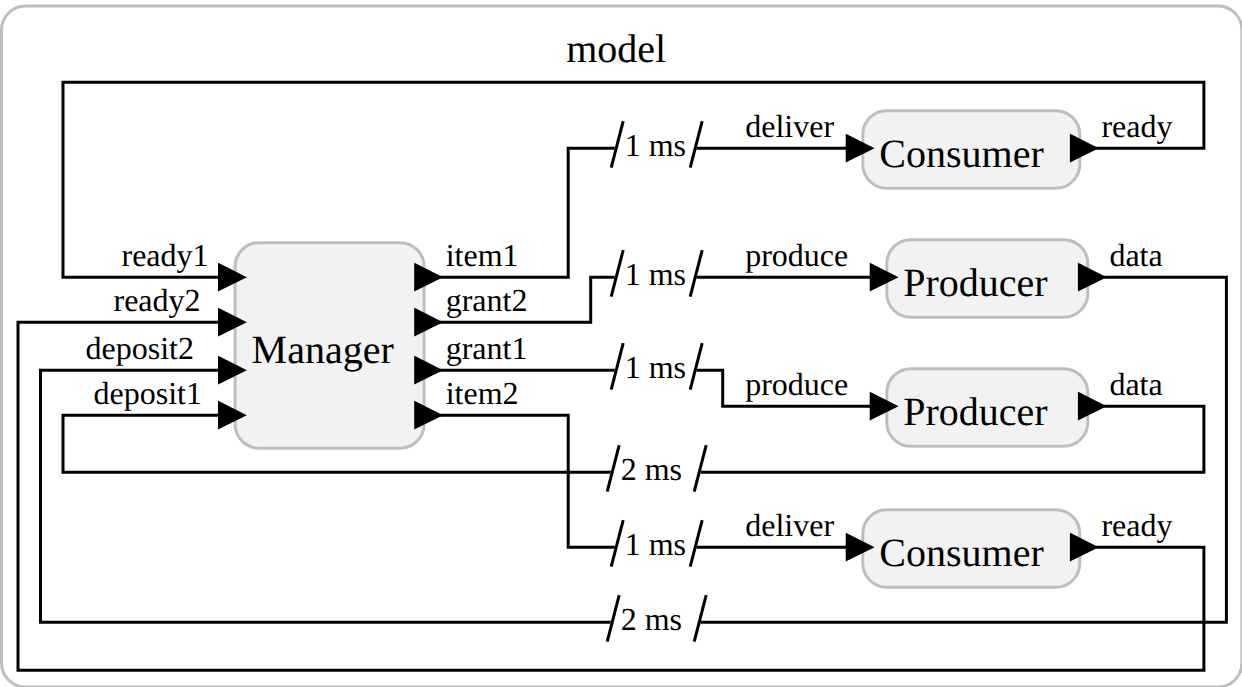}
\caption{Lingua Franca component diagram of the last benchmark sample. The
diagram exposes the producer and consumer message-dependency loops, but does not show
which loop can be affected by a code change.}
\label{fig:eval-causality-diagram}
\end{figure}

The two observable sets are
\[
\begin{aligned}
O_P &= \{\texttt{p1.produce},\texttt{p2.produce},\texttt{manager.deposit}\},\\
O_C &= \{\texttt{c1.deliver},\texttt{c2.deliver},\texttt{manager.consumerready}\}.
\end{aligned}
\]
Revision $M'_A$ changes the value delivered to consumers; $M'_B$ changes the
number of outstanding producer grants. Both retain the components and
connections. We also check a neutral revision $M'_N$ that replaces the argument
\texttt{last} by \texttt{last + 0} in a consumer delivery.
Table~\ref{tab:eval-causality} reports the static verdicts and whether the
derived twins change, including changes to observed message values.

\begin{table}[H]
\centering
\caption{Static verdicts and regenerated-twin comparisons for bounded-buffer
revisions. A rejected static comparison is inconclusive.}
\label{tab:eval-causality}
\small
\begin{tabular*}{\textwidth}{@{\extracolsep{\fill}}llcc@{}}
\hline
Revision & Observations & RDG slices & Derived twins \\
\hline
$M'_A$, consumer value & $O_P$ & equivalent & unchanged \\
$M'_A$, consumer value & $O_C$ & not equivalent & changed \\
$M'_B$, producer grants & $O_P$ & not equivalent & changed \\
$M'_B$, producer grants & $O_C$ & not equivalent & changed \\
$M'_N$, neutral expression & $O_C$ & not equivalent & unchanged \\
\hline
\end{tabular*}
\end{table}

For $M'_A$, slice equivalence establishes that the producer loop is unaffected,
whereas the consumer-loop comparison is rejected and its twin changes. For
$M'_B$, both comparisons are rejected and both twins change because producer
behavior can also influence later deliveries. The neutral revision $M'_N$ is
rejected despite its unchanged twin, showing that rejection is inconclusive.
Thus, the component diagram shows which causality loops exist, whereas the
RDG comparison determines from the code which loops are unaffected and which
may be affected by a particular evolution. This change-specific attribution
goes beyond analyses that stop at loop detection.

\section{Related Work}\label{sec:related}

Our static comparison uses syntactic equivalence of observable RDG slices as
a sufficient condition for weak timed bisimulation. A rejected comparison
indicates a possible effect, not necessarily a behavioral change.
Program dependence graphs support semantic comparison through the data and
control dependencies of source code~\cite{FerranteOW87}. Horwitz et
al.~\cite{HorwitzPR88a} show that isomorphic PDGs imply strong equivalence for
the sequential programs they consider. Ito et al.~\cite{ItoHagiharaYonezaki2007}
give an operational semantics for PDGs of unstructured programs and establish
its agreement with sequential execution. These results provide the foundation
for our comparison of message-server bodies. These sequential results do not
cover asynchronous message scheduling. Timed Rebeca additionally requires
preserving dependencies across message servers and the times at which messages
become available.
Other graph-based comparisons identify similar code. Krinke~\cite{Krinke2001}
detects similar subgraphs in PDGs, while Liu et al.~\cite{LiuZengWangLiang2023}
learn representations of code property graphs for functional similarity
detection. These methods support clone detection and code classification.
GumTree~\cite{FalleriEtAl2014} identifies source changes through edit scripts
between abstract syntax trees (ASTs), including node moves. Martinez et
al.~\cite{MartinezFalleriMonperrus2023} use hyperparameter optimization to
improve these scripts. Neither detected similarity nor an AST edit script
by itself establishes the behavioral equivalence required for reusing a tiny
twin.

Relational verification compares two programs through assertions relating
their executions. Benton's relational Hoare logic~\cite{Benton2004} supports
correctness proofs for static analyses and program transformations, including
slicing. Beckert et al.~\cite{BeckertEtAl2019} use relational verification to
check whether removing statements preserves the values selected by a slicing
criterion. Such methods can establish equivalence and can justify changes
beyond the syntactic matching used here. Their results concern the specified
relations between program states. Our criterion instead selects observable
message names, and our soundness result relates the induced timed transition
systems. Both approaches can reason from code without explicitly constructing
the complete state space.

Dependence analysis for concurrent and reactive languages must reflect their
communication semantics. Qi and Xu~\cite{QiXu2004} construct dependence graphs
from task-synchronization reachability graphs to analyze concurrent Ada
programs. Shankar and Fujita~\cite{ShankarFujita2008} derive equivalence rules
from the semantics of SpecC and apply them using dependence and flow graphs.
Their work is a direct precedent for establishing reactive-program
equivalence through a static graph representation.
Within the actor setting, Sabouri and Sirjani~\cite{SabouriSirjani2010Actor}
introduce RDGs for property-directed reduction of Rebeca models, including
static, step-wise, and bounded slicing. Sabouri and
Khosravi~\cite{SabouriKhosravi2012ProductLines} use slicing and static property
analysis to reduce the number of product configurations requiring model
checking. These approaches remove behavior irrelevant to a verification
task. We adapt the RDG to compare observable slices of two Timed Rebeca model
versions and establish a guarantee for their timed behavior. Actor-program
equivalence is also studied by Bereczky et al.~\cite{BereczkyHorpacsiThompson2024},
who formalize Core Erlang and use barbed bisimulation to relate programs with
different communication structures. Their result concerns an untimed actor
semantics and does not supply a criterion based on RDG-slice equivalence.

For tiny-twin reuse, detecting a causality cycle does not determine whether
a revision affects a selected observable loop.
Lingua Franca uses declared reaction dependencies to order execution and
detect causality problems~\cite{LohstrohEtAl2021}. Its coordination analysis
treats reaction bodies written in a target language as black boxes. The
feedback loops considered in our experiments can contain delayed messages;
they need not be instantaneous cycles that prevent a valid reaction order.
To determine whether a code revision can affect a selected loop, we follow
data, control, and message dependencies within and across message servers.
This is a static form of influence tracking, related to dependence-based
information-flow analysis~\cite{HammerSnelting2009} and semantic
differencing~\cite{Horwitz1990}. Choosing the
messages of each loop as its observable set allows unaffected loops to be
identified even when several loops share a component.

Abstraction offers another way to reduce verification cost. Counterexample-guided
abstraction refinement checks a smaller abstract model and refines it when
counterexamples are spurious~\cite{ClarkeEtAl2000}. Janowska and
Janowski~\cite{sliced} apply property-directed slicing to timed automata with
discrete data. Thieme et al.~\cite{DBLP:conf/zum/ThiemeSSPU26}
synthesize model slices that preserve cross-model consistency relationships,
using counterexample-guided inductive synthesis with Alloy. Their criterion
is equiconsistency with other models, whereas ours is preservation of
observable timed behavior under model evolution.

Tiny twins provide the immediate application of our comparison. Moradi et
al.~\cite{DBLP:journals/jpdc/MoradiPAS24} construct a tiny twin by generating
a Timed Rebeca state space, hiding non-observable actions, and reducing the
result under trace equivalence. The compact model is then used for runtime
attack detection. Our analysis addresses whether an existing tiny twin remains
valid after its source model changes. A successful comparison permits reuse
without generating and reducing the revised state space; it does not construct
the initial tiny twin.
Preserving elapsed time is essential to this reuse criterion. Weak timed
equivalences already accumulate durations across internal transitions, as
formalized by Brengos and Peressotti~\cite{BrengosPeressotti2019}. Building on
this notion, we connect a static comparison of observable RDG slices to weak timed
bisimulation, which preserves total elapsed time while hiding internal
actions. Weak timed bisimulation also implies weak timed trace equivalence
under the same message-name observations.
The result applies to the Zeno-free fragment with \texttt{after}
annotations and no \texttt{delay} statements.

\section{Conclusion}\label{sec:conclusion}
In this work, we have investigated the intersection of two classically separate views of programs: Reducing programs through slicing and program evolution. Our static analysis enables to characterize when slices are preserved under evolution of the program and evolution of the slicing context, i.e., the observables. 

The investigation is motivated by Digital Twins, whose longevity requires connecting program and model evolution with all operations performed on the model. We, thus, conjecture that static analysis can simplify Digital Twins reconfiguration in a broader setting. A core notion in this setting is \emph{consistency}~\cite{DBLP:conf/isola/PascualBUKP24,DBLP:journals/jss/MuctadirKCB26,DBLP:journals/sosym/KamburjanBJT26}: are the deployed models in a digital twin consistent with each other's requirements on the context. Our work can be understood as an analysis that establishes whether a changed context (observable or source model) affects the consistency of the deployed slice.

Static analyses can contribute to the investigation of Digital Twin evolution~\cite{DBLP:conf/models/AlskaifBBCCDHMNSBV25,DBLP:journals/sosym/MertensKBDH25,DBLP:journals/fgcs/AbbasiPS26} in two ways. First, as in this work, by enabling a precise notion of context changes and equivalences that make reconfiguration more efficient. Second, by providing a notion of context in the first place. In this work, we use observables (and the source model) as context of a slice, and a slice is consistent if its current context and the context it was generated under are equivalent for the slicing operation. Thus, a Digital Twin system managing consistency must make this context explicit for management.

\paragraph{Future Work}
Beyond integration with Digital Twin frameworks, we plan to investigate the evolution of slices under non-equivalence. Such methods would enable us to evolve the slice instead (re-)generating it and reuse information from prior extractions.

\section*{Acknowledgments}
This work is partially supported by the DFF project \textit{Graph-based Verification of Reflective Programs} (5254-00016B).








\bibliographystyle{elsarticle-num}
\bibliography{ref}





\end{document}